\documentclass{cmslatex}
\usepackage[paperwidth=7in, paperheight=10in, margin=.875in]{geometry}
 \usepackage[backref,colorlinks,linkcolor=red,anchorcolor=green,citecolor=blue]{hyperref}
\usepackage{amsfonts,amssymb}
\usepackage{amsmath}
\usepackage{algpseudocode,algorithm,algorithmicx}
\usepackage{bm}
\usepackage{epsfig} 
\usepackage{graphicx}
\usepackage{cite}
\usepackage{enumerate}
\usepackage{comment}

\renewcommand{\theequation}{\arabic{section}.\arabic{equation}}

\DeclareMathOperator{\trace}{trace}

\def\la{\left\langle}
\def\ra{\right\rangle}
\def\lb{\left(}
\def\rb{\right)}

\newcommand{\matsnorm}[2]{\left\| #1\right\|_{{#2}}}

\newcommand{\fronorm}[1]{\ensuremath{\matsnorm{#1}{\footnotesize{\mathsf{F}}}}}
\newcommand{\opnorm}[1]{\ensuremath{\matsnorm{#1}{}}}

\newcommand{\twonorm}[1]{\ensuremath{\matsnorm{#1}{\footnotesize{2}}}}

\newcommand{\bfm}[1]{\bm{#1}}

\newcommand{\E}[2][]{\mathbb{E}_{#1} \left[ #2 \rule{0mm}{3mm}\right]}

\def\va{\bfm a}   \def\mA{\bfm A}  
\def\vb{\bfm b}   \def\mB{\bfm B}  
     \def\C{\mathbb{C}}
     
\def\ve{\bfm e}     
    
\newcommand\vg{\bfm g}     
\def\vh{\bfm h}     
   \def\mI{\bfm I}

   \def\mM{\bfm M}  
   \def\mN{\bfm N}

     \def\R{\mathbb{R}}

   \def\mU{\bfm U}  
   \def\mV{\bfm V}  
   \def\mW{\bfm W}  
\def\vx{\bfm x}   \def\mX{\bfm X}  
\def\vy{\bfm y}     
\def\vz{\bfm z}   \def\mZ{\bfm Z}

\def\calA{{\cal  A}}

\def\calD{{\cal  D}}

\def\calG{{\cal  G}} 
\def\calH{{\cal  H}} 
\def\calI{{\cal  I}}

\def\calO{{\cal  O}} 
\def\calP{{\cal  P}}

\def\calT{{\cal  T}}

\def\calW{{\cal  W}}

\newcommand{\bfsym}[1]{\bm{#1}}

             \def\bSigma{\bfsym \Sigma}

\newcommand{\bX}{\boldsymbol{X}}

\newcommand{\bB}{\boldsymbol{B}}
\newcommand{\bC}{\mathbb{C}}

\newcommand{\bb}{\boldsymbol{b}}

\newcommand{\be}{\boldsymbol{e}}

\newcommand{\bZ}{\boldsymbol{Z}}
\newcommand{\bW}{\boldsymbol{W}}
\newcommand{\bh}{\boldsymbol{h}}

\newcommand{\by}{\boldsymbol{y}}

\def \calGT{\calG^{\ast}}

\def \calAT{\calA^{\ast}}

\def \tran {\mathsf{H}}
\def \tranH{\mathsf{H}}

\def \bone{\bm 1}

   \allowdisplaybreaks
\begin{document}
  \title{Blind Super-resolution of Point Sources via  Fast Iterative Hard Thresholding\thanks{Received date, and accepted date (The correct dates will be entered by the editor).}}
 
 
 \author{Zengying Zhu\thanks{Z.~Zhu is with School of Mathematical Sciences, Fudan University, Shanghai, China, zengyingzhu@fudan.edu.cn}
 	\and
 	Jinchi Chen \thanks{J.~Chen is with School of Data Science, Fudan University, Shanghai, China, e-mail: jcchen.phys@gmail.com. This research is supported by NSFC 12001108}
 	\and Weiguo Gao\thanks{W.~Gao is with School of Mathematical Sciences and School of Data Science, Fudan University, Shanghai, China, e-mail: wggao@fudan.edu.cn. This research is supported in part by NSFC 11690013 and 71991471.}}

 \pagestyle{myheadings} \markboth{BLIND SUPER-RESOLUTION OF POINT SOURCES VIA  FIHT-VHL}{Z.~ZHU, J.~CHEN AND W.~GAO} \maketitle
 
 \begin{abstract} In this work, we develop a provable fast algorithm for blind super-resolution based on the low rank structure of vectorized Hankel matrix associated with the target matrix. Theoretical results show that the proposed method converges to the ground truth with linear convergence rate. 
 	Numerical experiments are also conducted to illustrate the linear convergence and effectiveness of the proposed approach.
 \end{abstract}

\section{Introduction}

Blind super-resolution is the problem of estimating $\{\tau_k, d_k, \vg_k\}_{k=1}^r$ from the observations
\begin{align}
	\label{eq: samples}
	\vy[j] = \sum_{k=1}^{r} d_k e^{-2\pi\imath \tau_k(j-1)}\vg_k[j],\text{ for }j=1,\ldots, n,
\end{align}
where $\{d_k, \tau_k\}_{k=1}^r$ are unknown coefﬁcients and frequencies associated with the complex exponentials, and $\{\vg_k\}$ are unknown point spread functions. It has many applications, such as computational photography \cite{fergus2006removing}, multi-user communication system \cite{luo2006low} and seismic data analysis \cite{margrave2011gabor}. It is worth noting that the measurement model \eqref{eq: samples} also includes blind sparse spike deconvolution problem \cite{chi2016guaranteed,li2016rapid} in which the point spread function is shared among all point sources.  In particular, when the knowledge of $\{\vg_k\}_{k=1}^r$ is available, blind super-resolution reduces to the super-resolution problem \cite{Armin18Sparse,Bernstein2018DeconvolutionOP}.

Since the number of unknowns in \eqref{eq: samples} is larger than the number of samples, it is an ill-posed problem without any additional constraints on $\vg_k$. To alleviate this issue, it is typically assumed that $\{\vg_k\}_{k=1}^r$ belong to a known low-dimensional subspace spanned by the columns of $\mB\in\C^{n\times s}$ with $s<n$, i.e., 
\begin{align}
	\label{eq: subspace assumption}
	\vg_k = \mB\vh_k,
\end{align}
where $\vh_k\in\C^s$ represents the unknown coefficient of $\vg_k$ in the subspace \cite{chi2016guaranteed, yang2016super, chen2020vectorized}. Under the subspace assumption \eqref{eq: subspace assumption} and applying the lift technique \cite{ahmed2013blind,li2016rapid}, blind super-resolution can be cast as the problem of recovering the matrix $\mX^\natural = \sum_{k=1}^r d_k \vh_k \va_{\tau_k}^\tran\in\C^{s\times n}$ from a set of linear measurements
\begin{align}
	\label{eq: measurements}
	\vy[j] = \la \vb_j\ve_j^\tran, \mX^\natural \ra,\quad j=0,\ldots, n-1, 
\end{align}
where $\va_\tau = \begin{bmatrix}
	1 & e^{-2\pi \imath \tau\cdot 1} &\cdots & e^{-2\pi \imath \tau\cdot (n-1)}
\end{bmatrix}^\tran$, $\vb_j\in\C^s$ is the $j$-th column of $\mB^\tranH$, $\ve_j$ is the $(j+1)$-th standard basis of $\R^n$, and the inner product between two matrices is given by $\la\mA,\mB\ra = \trace(\mA^\tranH\mB)$. The measurement model \eqref{eq: measurements} can be rewritten as a more compact form 
\begin{align}
	\label{eq: linear measurements}
	\vy = \calA(\mX^\natural),
\end{align}
where $\calA:\C^{s\times n}\rightarrow \C^m$ is the linear operator.  When the data matrix $\mX^\natural$ is recovered, the frequencies $\{ \tau_k\}_{k=1}^r$ can be extracted from it via spatial smoothing MUSIC \cite{Evans1982ApplicationOA} and coefficients $\{d_k, \vh_k\}_{k=1}^r$ can be estimated by solving an over--determined linear system \cite{yang2016super}. Therefore in this work we focus on the problem of estimating $\mX^\natural$ from its linear measurements \eqref{eq: linear measurements}.  

Let $\calH$ be the vectorized Hankel lift operator \cite{zhang2018multichannel,chen2020vectorized} which maps a matrix $\mX\in\C^{s\times n}$ into an $sn_1\times n_2$ matrix,
\begin{align}\label{def:vhl}
	\calH(\mX)= \begin{bmatrix}
		\vx_1& \vx_{2} &\cdots &\vx_{n_2}\\
		\vx_{2} & \vx_{3} & \cdots & \vx_{n_{2}+1} \\
		\vdots &\vdots &\ddots &\vdots\\
		\vx_{n_1} &\vx_{n_1+1} &\cdots &\vx_{n},
	\end{bmatrix}\in \mathbb{C}^{s n_{1} \times n_{2}}
\end{align}
and $\vx_i\in\C^s$ is the $i$-th column of $\mX$ and $n_1+n_2 = n+1$. It has been shown that the rank of $\calH(\mX^\natural)$ is at most $r$ and thus it is a low rank matrix when $r\ll \min(sn_1, n_2)$ \cite{chen2020vectorized}. To recover $\mX^\natural$, we seek a rank-$r$ vectorized Hankel matrix consistent with the linear measurements \eqref{eq: linear measurements} by solving the following low rank vectorized Hankel matrix sensing problem
\begin{align}
\label{eq: nonconvex recovery procedure}
\min_{\mX\in\C^{s\times n}} \frac{1}{2} \twonorm{\vy - \calA(\mX)}^2\text{ s.t. }\rank(\calH(\mX)) = r,
\end{align}

Inspired by \cite{CAI2019Fast, zhang2018multichannel}, we develop a non-convex algorithm called Fast Iterative Hard Thresholding via Vectorized Hankel Lift (FIHT-VHL) to solve \eqref{eq: nonconvex recovery procedure}. We apply the low rank structure of the vectorized Hankel matrix associated with the target matrix while Cai et.al.~\cite{CAI2019Fast} investigated the low rank Hankel matrix recovery problem in the context of spectrally sparse recovery. Moreover, the measurement model in this work is different from that in \cite{zhang2018multichannel}. Therefore the theoretical guarantee in \cite{zhang2018multichannel} is not applicable for the blind super-resolution
setting. We show that FIHT-VHL is able to converge linearly to the unknown data matrix with high probability if the number of measurements is of the order $\calO(s^2 r^2\log^2(sn))$ and the algorithm is properly initialized.

\textbf{Related works:} 
The problem of recovering $\mX^\natural$ from \eqref{eq: nonconvex recovery procedure} is also studied in \cite{chen2020vectorized, mao2021projected}. 
Chen et.al.~\cite{chen2020vectorized} developed an nuclear norm minimization method based on the  vectorized Hankel lift. Recently, Mao and Chen \cite{mao2021projected} developed a Projected Gradient Descent via Vectorized Hankel Lift (PGD-VHL) method for the problem \eqref{eq: nonconvex recovery procedure}. Their theoretical results show that the matrix $\mX^\natural$ can be exactly recovered from $\calO(s^2 r^2\log(sn))$ measurements.

Another line of related work addresses the problem of recovering $\mX^\natural$ from \eqref{eq: linear measurements}  \cite{yang2016super,  li2019atomic, suliman2021mathematical}. More specifically, Yang et.al. \cite{yang2016super} proposed an atomic norm minimization method to recover the data matrix. Their theoretical result shows that $\calO(sr\log n)$ measurements are sufficient to guarantee exact recovery of $\mX^\natural$ with high probability under certain incoherence condition. The stable analysis of blind super-resolution is also provided in \cite{li2019atomic,suliman2021mathematical}.  The approaches developed in  \cite{yang2016super,  li2019atomic, suliman2021mathematical} are based on convex relaxation and the equivalent semi-definite programmings are computational inefficient for large-scale problems.






\textbf{Organization:} The remainder of this paper is organized as follows. In Section 2, we will introduce FIHT-VHL algorithm. In Section 3, we will introduce two assumptions and establish our main result. The performance of FIHT-VHL is evaluated by numerical experiments in Section 4. In Section 5, we give the detailed proofs for main result. We close with a conclusion in Section 6.

\textbf{Notations and preliminaries:} Throughout this work, we use bold lowercase letters, bold uppercase letters and calligraphic letters for vectors, matrices and operators, respectively. The letter $\calI$ denotes the identity operator. We use $\vx[i]$ to denote the $i$-th entry of vector $\vx$ and $\mX[j,k]$ to denote the $(j,k)$-th entry of matrix $\mX$. Additionally, we use $\mZ[i:j, k]$ to denote a $j-i+1$ vector with entries $\mZ[i,k],\ldots, \mZ[j,k]$. The adjoint of $\mathcal{H}$, denoted by $\mathcal{H}^{*}$, is a linear mapping from $s n_{1} \times n_{2}$ matrices to matrices of size $s \times n$. In particular, for any matrix $\bZ \in \mathbb{C}^{s n_{1} \times n_{2}}$, the $i$-th column of $\mathcal{H}^{*}(\boldsymbol{Z})$ is given by
\begin{align*}
\calH^\ast(\mZ) \ve_{i}=
\sum_{(j,k)\in\calW_i }\vz_{j, k},
\end{align*} where $\vz_{j,k}= \mZ[js:(j+1)s-1, k]$ and $\calW_i$ is the set \begin{align*}
\left\{(j,k)\mid j+k=i, 0 \leq j \leq n_{1}-1,0 \leq k \leq n_{2}-1\right\}.\end{align*}
Let $\calD:\C^{s\times n}\rightarrow \C^{s\times n}$ be an operator such that $$\calD(\mX)= \mX \diag\left(\sqrt{w_{0} },\ldots, \sqrt{w_{n-1}}\right)$$ for any $\mX$, where the scalar $w_{i}$ is defined as the number of $\calW_i$ for $i=0, \ldots, n-1$. The Moore-Penrose pseudoinverse of $\mathcal{H}$ is given by $\mathcal{H}^{\dagger}=\mathcal{D}^{-2} \mathcal{H}^{*}$ which satisfies $\mathcal{H}^{\dagger} \mathcal{H}=\mathcal{I}$ \cite{chen2020vectorized}.
The adjoint of the operator
$\mathcal{A}(\cdot)$, denoted by  $\mathcal{A}^{*}(\cdot)$, is defined as $\mathcal{A}^{*}(\boldsymbol{y})=\sum_{j=0}^{n-1} \boldsymbol{y}[j] \boldsymbol{b}_{j} \boldsymbol{e}_{j}^\tranH$. Denote $\bZ = \calH{(\bX)}$ and  $\mathcal{G}=\mathcal{H} \mathcal{D}^{-1}$. The adjoint of $\mathcal{G}$, denoted by $\mathcal{G}^{*}$, is given by $\mathcal{G}^{*}=\mathcal{D}^{-1} \mathcal{H}^{*}$. 

Let $\mZ = \mU \bSigma {\mV}^\tranH\in\C^{sn_1\times n_2}$ be the compact singular value decomposition of a rank-$r$ matrix, where $\mU\in\C^{sn_1\times r}$, $\mV \in\C^{n_2\times r}$ and $\bSigma\in\R^{r\times r}$. It is known that the tangent space of the fixed rank-$r$ matrix manifold at $\mZ$ is given by \cite{Vandereycken2013LowRankMC}
\begin{align*}
	\mathfrak{T} = \left\{ \mU\mN^\tranH + \mM{\mV}^\tranH: \mM\in\C^{sn_1\times r}, \mN\in\C^{n_2\times r}\right \}.
\end{align*}
Given any matrix $\mW\in\C^{sn_1\times n_2}$, the projection of $\mW$ onto $\mathfrak{T}$ can be computed using the formula \cite{Vandereycken2013LowRankMC}
\begin{align*}
	\calP_{\mathfrak{T}}(\mW) = \mU{\mU}^\tranH\mW + \mW\mV{\mV}^\tranH - \mU{\mU}^\tranH\mW \mV{\mV}^\tranH.
\end{align*}

\section{Fast Iterative Hard Thresholding via Vectorized Hankel Lift}

We develop a fast iterative hard thresholding method for the problem \eqref{eq: nonconvex recovery procedure}, which is summarized in Algorithm $\bone$. 
The initial guess was obtained with the spectrum method.
In the $t$-th iteration of FIHT-VHL, the current estimate $\bX^t$ is first updated along the gradient descent direction of the objective in \eqref{eq: nonconvex recovery procedure}. Then, the vectorized Hankel matrix corresponding to the update is formed via the application of the vectorized Hankel lift operator $\calH$, followed by a projection operator $\calP_{\mathfrak{T}_t}$ onto the $\mathfrak{T}_t$ space. 
After that, it imposes a hard thresholding operator $\mathcal{T}_{r}$ to $\bW^t$ by truncated SVD process. 
Finally, it applies $\calH^\dagger$ on the low rank matrix $\bZ^{t+1}$.
Indeed, FIHT-VHL algorithm can be efficiently implemented. The authors in \cite{chen2020vectorized} show that $\calH(\mX)$ and $\calH^\dagger(\mZ)$ can be computed by using $\calO(srn\log n)$ flops. Moreover, instead of computing a SVD directly, FIHT-VHL first projects a matrix onto a $2r$-dimensional subspace and then calculates the SVD of a rank-$2r$ matrix, which requires $\calO(r^2sn + r^3)$ flops. Thus the main computational complexity in each step is $\calO(r^2sn+r^3 + srn\log n)$ and therefore
our algorithm is very efficient compared to existing algorithms.

\begin{algorithm}[htbp]\label{algo:fiht}
\caption{FIHT-VHL}  
	\hspace*{0.02in} {\bf Input:} Initialization $\mX^0 = \calH^\dagger\calT_r\calH\calA^\ast(\vy)$.\\
	\hspace*{0.02in} {\bf Output:} $\mX^T$

	\begin{algorithmic}[1]
		\For{$t = 0,1,\dots,T-1$}
		\State $\widetilde{\bX}^t= \bX^t- \alpha \calA^*(\calA(\bX^t)-\by)\label{eq:Xtilde}$ 
		\State $\bW^t =\calP_{\mathfrak{T}_t} \calH(\widetilde{\bX}^t)$
		\State $\bZ^{t+1}=\mathcal{T}_{r}(\bW^t)$
		\State $\bX^{t+1}= \calH^{\dagger}(\bZ^{t+1})$
        \EndFor
	\end{algorithmic}
\end{algorithm}

\section{Main Result}\label{section:main result}

In this section, we establish our main result. To this end, we make two assumptions.

\begin{assumption}\label{assump:1}
 The column vectors $\left\{\boldsymbol{b}_{j}\right\}_{j=0}^{n-1}\subset \C^s$ of the subspace matrix $\boldsymbol{B}^{\tranH}$ are i.i.d random vectors which obey 
 \begin{align*}
    \E{\vb_j\vb_j^\tranH} = \mI_s \text{ and } \max _{0 \leq \ell \leq s-1}|\vb_j[\ell]|^{2} \leq \mu_{0},
 \end{align*}
 for some constant $\mu_0$. Here, $\vb_j[\ell]$ denotes the $\ell$-th entry of $\vb_j$.
\end{assumption}
\begin{remark}
This assumption is standard in compressed sensing \cite{candes2011probabilistic} and blind super-resolution \cite{chi2016guaranteed, yang2016super,li2019atomic,chen2020vectorized}, and holds with $\mu_0=1$ when $\vb$ is uniformly sampled from the rows of a Discrete Fourier Transform (DFT) matrix.
\end{remark}

\begin{assumption}\label{assump:2}
There exists a constant $\mu_{1}>0$ such that
\begin{align*}
\max_{0\leq i \leq n_1-1} \fronorm{\mU_i}^2\leq \frac{\mu_1 r}{n} \text{ and } \max_{0\leq j \leq  n_2-1} \twonorm{\ve_j^\tran \mV}^2\leq \frac{\mu_1 r}{n},
\end{align*}
where the columns of $\mU\in\C^{sn_1\times r}$ and $\mV\in\C^{n_2\times r}$ are the left and right singular vectors of $\mZ^\natural=\calH(\mX^\natural)$ separately, and $\mU_i = \mU[is:(i+1)s-1]$ is the $i$-th block of $\mU$.
\end{assumption}
\begin{remark}
Assumption \ref{assump:2} is commonly used in spectrally sparse signal recovery \cite{cai2015fast,cai2018spectral,chen2014robust} and blind super-resolution \cite{chen2020vectorized,mao2021projected}, and is satisfied when the minimum separation distance between $\{\tau_k\}_{k=1}^r$ is greater than about $1/n$. 
\end{remark}

Now, we are in the position to state our main result, whose proof is deferred to Section \ref{section:proof}.
\begin{theorem}
\label{thm 1}
Under Assumption \ref{assump:1} and \ref{assump:2}, with probability at least $1- c_1n^{-c_2}$, the iterations generated by FIHT-VHL with the initial guess 
$\bX^0 = \calH^\dagger\mathcal{T}_{r} \calH (\calA^*(\by))$ satisfies 
\begin{align}
	\fronorm{\mX^{t} - \mX^\natural} \leq \left( \frac{1}{2} \right)^t \fronorm{\mX^0 - \mX^\natural}
\end{align}
provided the step size $\alpha=1$ and the sample complexity obeys that
\begin{align*}
n\geq   C\kappa^2\mu_0^2\mu_1s^2r^2\log^2(sn)/\varepsilon^2
\end{align*}
where $c_1,c_2$ and $C$ are absolute constants and $\kappa = \sigma_{\max}(\calH{(\mX^\natural)}) / \sigma_{\min}(\calH{(\mX^\natural)})$.
\end{theorem}

\begin{remark}
The sample complexity established in \cite{chen2020vectorized} for the Vectorized Hankel Lift is $n\geq c\mu_0 \mu_1\cdot sr\log^4(sn)$. While the sample complexity is sub-optimal dependence on $s$ and $r$, our recovery method requires low per iteration computational complexity. 
\end{remark}
\begin{remark}
It is shown in \cite{mao2021projected} that PGD-VHL can exact recover $\mX^\natural$ when the measurements is of order $\calO(s^2r^2\log^2 (sn))$. Moreover the exact recovery of PGD-VHL relies on a more complicated regularization scheme. 
\end{remark}


\section{Numerical Simulations}

Numerical experiments are conducted to evaluate the performance of FIHT-VHL. In the experiments, the target matrix $\mX^\natural$ is generated by$\mX^\natural = \sum_{k=1}^r d_k \vh_k \va_{\tau_k}^\tran$ and the measurements are obtained by \eqref{eq: samples}, where the locations $\{\tau_{k}\}_{k=1}^r$ are generated from a standard uniform distribution  $U(0,1)$, and the amplitudes $\{d_k\}_{k=1}^r$ are generated via $d_k = (1+10^{c_k})e^{-i\psi_k}$ where $\psi_k$ follows $U(0,2\pi)$ and $c_k$ follows $U(0,1)$. Each row of subspace matrix $\mB$ is uniformly sampled from the rows of a Discrete Fourier Transform matrix. The coefficient vectors $\{\bh_k\}_{k=1}^r$ are generated from a standardized multivariate Gaussian distribution  $MVN_s(0,I_{s\times s})$, where $I_{s\times s}$ is the identity matrix.


In the first experiment, we study the convergence rate of FIHT-VHL under different number of observations and compare it with PGD-VHL \cite{mao2021projected}. The step size for FIHT-VHL is given by 
\begin{equation}\label{eq: step size}
\alpha = \frac{\|\calP_{T_t}\calG(\by - \calA(\bX^t))\|_2^2}{\|\calA\calG^*\calP_{T_t}(\by - \calA(\bX^t))\|_2^2}
\end{equation}
in each iteration and for PGD-VHL, the step size is chosen using the line search method. 
We set dimensions of subspaces $s = 4$ and the number of point sources $r = 4$. Fig.~\ref{fig:convergence123} presents the logarithmic recovery error $\log_{10}\fronorm{\mX_t - \mX^\natural}/\fronorm{\mX^\natural}$ with respect to the number of iterations. Fig.~\ref{fig:convergence123} shows that FIHT-VHL converges linearly which is in accordance with our main theorem. Compared to PGD-VHL \cite{mao2021projected}, FIHT-VHL requires fewer number of iterations to achieve convergence when $s=r=4$.  

\begin{figure}[htbp]
\begin{center}
\includegraphics[width=0.55\textwidth]{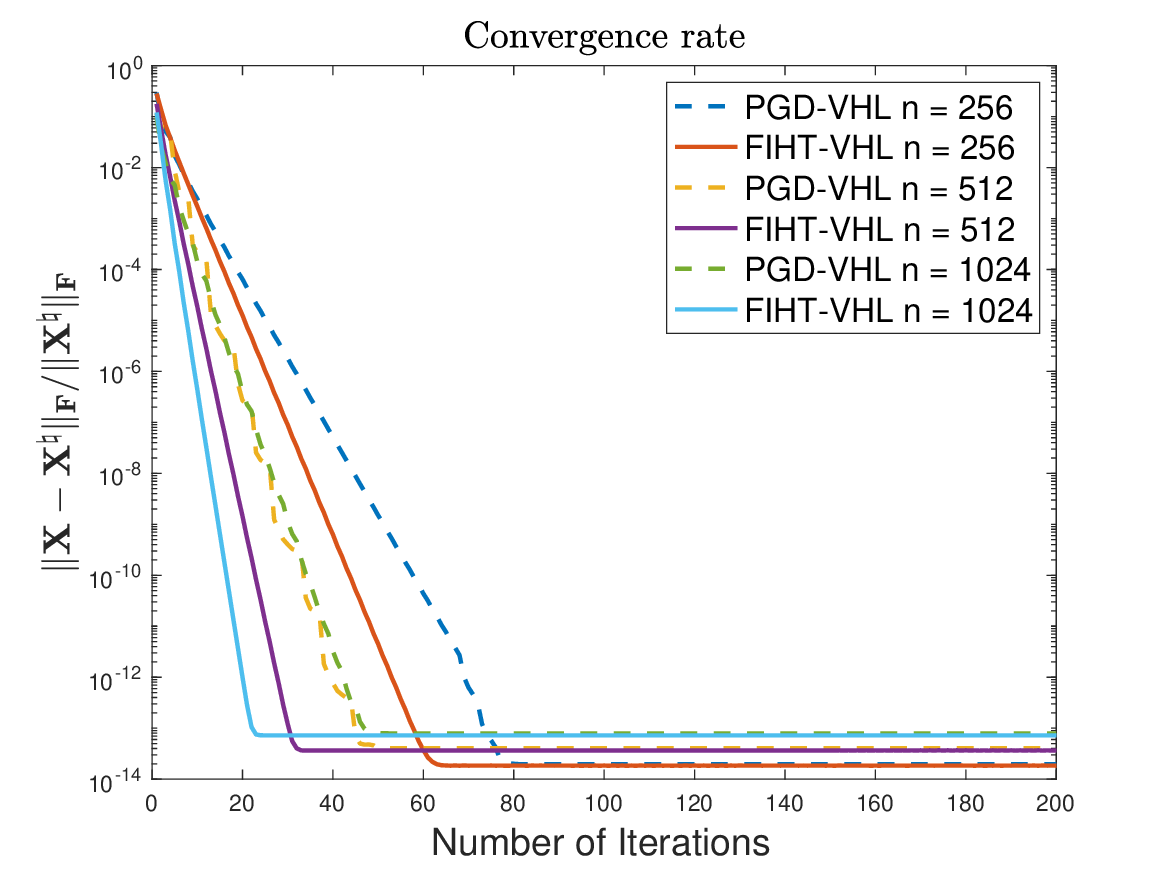} 
\caption{Convergences of FIHT-VHL and PGD-VHL for $n=256, 512, 1024$ when $s=4$ and $r=4$.}
\label{fig:convergence123}
\end{center}
\end{figure}

In the second experiment, we conduct tests to illustrate the robustness of FIHT-VHL to additive noise. More specifically, we add noise vector $\be = \sigma_{\be}\cdot \|\by\|_2\cdot\frac{\boldsymbol{w}}{\|\boldsymbol{w}\|_2}$ to the measurements where $\by$ is the noiseless observations \eqref{eq: measurements}, $\sigma_{\be}$ denotes the noise level and  $\boldsymbol{w}$ is the standard Gaussian vector with i.i.d entries. In the tests, the noise level $\sigma_{\be}$ is evenly spaced from $10^{-5}$ to $10^{-3}$, corresponding to the signal-to-noise ratio (SNR) from 100 to 60 dB. For each noise level, 10 random trails are conducted with $s=r=2$. We choose $n = 128$ and $ n = 256$ for the number of measurements. In Fig.~\ref{fig:robustness}, we demonstrate  the linear relationship between the average relative reconstruction error and the noise level. It can be seen that the relative recovery error decreases with the
increase of the number of measurements.

\begin{figure}[htbp]
\begin{center}
\includegraphics[width=0.55\textwidth]{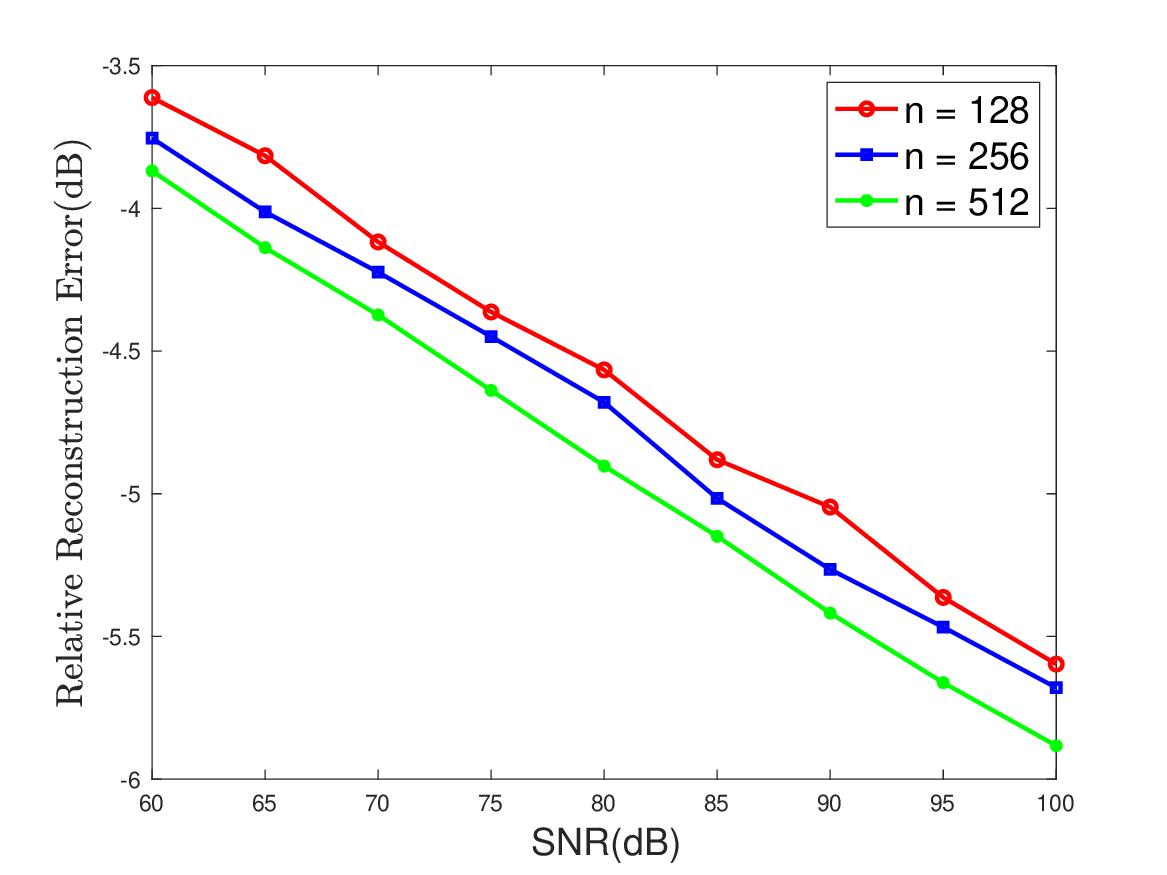} 
\caption{Performance of FIHT-VHL under different noise levels when $s=r=2$.}
\label{fig:robustness}
\end{center}
\end{figure}


\section{Proof of Main Result}\label{section:proof}

We first introduce three auxiliary lemmas that will be used in our proof.

\begin{lemma}[{\cite[Corollary III.9]{chen2020vectorized}}]
\label{lemma 5.1}
Suppose $n\geq C\varepsilon^{-2}  \mu_0 \mu_1 sr\log(sn)$. The event
\begin{align}
	\label{ineq: local rip}
	\opnorm{ \calP_{T} \lb \calG\calGT -  \calG\calAT \calA \calGT \rb \calP_{T}  }  \leq \varepsilon
\end{align} 
occurs with probability at least $1-c_1n^{-c_2}$.
\end{lemma}

\begin{lemma}[{\cite[Lemma 5.2]{mao2021projected}}]
\label{lemma initial}
Suppose that $n\geq C\varepsilon^{-2}\kappa^2 \mu_0^2\mu_1 s^2 r^2\log^2(sn)$. Then with probability at least $1-c_1n^{-c_2}$, the initialization $\mZ_0 = \calH(\mX^0)$ obeys 
\begin{align*}
\opnorm{\mZ^0 - \mZ^\natural}\leq \frac{\sigma_{r}(\mZ^\natural)\varepsilon}{16\sqrt{(1+\varepsilon) \mu_0s}}.
\end{align*}
\end{lemma}

\begin{lemma}\label{lemma:1}
Suppose that
\begin{align}
	\label{condition 1}
\fronorm{\mZ^t - \mZ^\natural}  \leq \frac{\sigma_{r}(\mZ^\natural)\varepsilon}{16\sqrt{(1+\varepsilon)\cdot \mu_0 s}}.
\end{align}
Conditioned on \eqref{ineq: local rip}, one has
\begin{align}
	\label{ineq: one side rip}
	\opnorm{ \calA \calGT\calP_{\mathfrak{T}_t}}& \leq 3\sqrt{ 1+\varepsilon},\\
	\label{ineq: local rip at t step}
	\opnorm{\calP_{\mathfrak{T}_t}\calG\lb  \calI - \calA^\ast\calA\rb\calGT\calP_{\mathfrak{T}_t} } &\leq 2\varepsilon.
\end{align}
\end{lemma}

Now we are in the position to prove our main result. We rewrite the iteration in FIHT-VHL as
\begin{align}
	\mZ^{t+1}= \calT_r \lb \mZ^t -\calG\calAT \calA \calGT\lb \mZ^t - \mZ^\natural \rb\rb .
\end{align}
For ease of exposition, we will prove lemmas and theorems in this section in terms of $\bZ^{t}$ and $\bZ^{\natural}$ but note that results in terms of $\bX^{t}$ and $\bX^{\natural}$ follows immediately due to
\begin{align}
\fronorm{\mX^{t}-\mX^{\natural}} = \fronorm{\calH^\dagger(\mZ^{t}-\mZ^{\natural})}\leq \fronorm{\mZ^{t}-\mZ^{\natural}}.
\end{align}
Motivated by \cite{CAI2019Fast}, we prove our main result by induction. 
When $t=0$, the condition \eqref{condition 1} holds by Lemma \ref{lemma initial}. Assume in the $t$-th step, linear convergence holds, we need to prove that in $(t+1)$-th step, linear convergence also holds.  

Denote 
$$\mW^t = \calP_{\mathfrak{T}_t} \lb \mZ^t -  \calG\calAT \calA \calGT\lb \mZ^t - \mZ^\natural \rb \rb.$$
We have that $\mZ^{t+1} = \calT_r(\mW^t)$ and  $\fronorm{\mZ^{t+1}- \mZ^\natural} $ can be bounded as follows:
\begin{align*}    
\fronorm{\mZ^{t+1}- \mZ^\natural} & \leq \fronorm{\mZ^{t+1} - \mW^t } +\fronorm{\mW^t - \mZ^\natural}\notag\\
& \leq 2\fronorm{ \calP_{\mathfrak{T}_t} \lb \mZ^t -\mZ^\natural - \calG\calAT \calA \calGT\lb \mZ^t - \mZ^\natural \rb \rb }  + 2\fronorm{ \lb \calI -  \calP_{\mathfrak{T}_t}  \rb \lb \mZ^\natural \rb  } \\
 	& = 2\big\| \calP_{\mathfrak{T}_t} \big( (1-\eta)\lb \mZ^t -\mZ^\natural \rb + \eta \lb \mZ^t -\mZ^\natural  \rb -\eta\calG\calAT \calA \calGT\lb \mZ^t - \mZ^\natural \rb \big) \big\|_{\mathsf{F}}\\
 	&\quad + 2\big\| \lb \calI -  \calP_{\mathfrak{T}_t}  \rb \lb \mZ^\natural \rb  \big\|_{\mathsf{F}} \\
 	& \leq 2(1-\eta) \fronorm{ \calP_{\mathfrak{T}_t} \lb \mZ^t -\mZ^\natural \rb }+ 2\eta \fronorm{ \calP_{\mathfrak{T}_t} \lb \calG\calGT - \calG\calAT \calA \calGT \rb \lb \mZ^t - \mZ^\natural \rb  } \\
 	&\quad + 2\fronorm{ \lb \calI -  \calP_{\mathfrak{T}_t}  \rb \lb \mZ^\natural \rb  } \\
	& \leq2\fronorm{ \lb \calI -  \calP_{\mathfrak{T}_t}  \rb \lb \mZ^\natural \rb  } +2 \fronorm{ \calP_{\mathfrak{T}_t} \lb \calG\calGT - \calG\calAT \calA \calGT \rb \calP_{\mathfrak{T}_t}  \lb \mZ^t - \mZ^\natural \rb  }  \\
	&  \quad+  2 \fronorm{ \calP_{\mathfrak{T}_t} \lb \calG\calGT -  \calG\calAT \calA \calGT \rb \lb \calI -  \calP_{\mathfrak{T}_t} \rb \lb \mZ^t - \mZ^\natural \rb  } \\
	& \leq 2\fronorm{ \lb \calI -  \calP_{\mathfrak{T}_t}  \rb \lb \mZ^t-\mZ^\natural \rb }+2  \fronorm{ \calP_{\mathfrak{T}_t} \lb \calG\calGT -  \calG\calAT \calA \calGT \rb \calP_{\mathfrak{T}_t}  \lb \mZ^t - \mZ^\natural \rb}\\
	& \quad+2  \fronorm{ \calP_{\mathfrak{T}_t}   \calG\calGT  \lb \calI -  \calP_{\mathfrak{T}_t} \rb \lb \mZ^t - \mZ^\natural \rb  } \\
	&\quad +2  \fronorm{\calP_{\mathfrak{T}_t} \calG\calAT\calA\calGT \lb \calI -  \calP_{\mathfrak{T}_t} \rb\lb \mZ^t - \mZ^\natural \rb  }\\
	& \triangleq I_1+I_2+I_3+I_4.
\end{align*}
We bound $I_1,I_2,I_3$ and $I_4$, respectively. Applying \cite[Lemma 4.1]{wei2016guarantees} yields that
\begin{align*}
I_1 \leq 2\fronorm{\mZ^t - \mZ^\natural} ^2/ \sigma_{r}(\mZ^\natural), \quad I_3 \leq 2\fronorm{\mZ^t - \mZ^\natural} ^2/ \sigma_{r}(\mZ^\natural)
\end{align*}
For $I_2$, a simple computation yields that
\begin{align*}
    I_2 &\leq 2 \opnorm{\calP_{\mathfrak{T}_t} \lb \calG\calGT -  \calG\calAT \calA \calGT \rb \calP_{\mathfrak{T}_t} } \cdot \fronorm{ \mZ^t - \mZ^\natural}\\
    & \leq 2\varepsilon \fronorm{\mZ^t - \mZ^\natural}
\end{align*}
where the last line is due to Lemma \ref{lemma 5.1}. Finally, Lemma \ref{lemma:1} implies that 
\begin{align*}
I_4 \leq 3 \sqrt{1+\varepsilon}\cdot \fronorm{\mZ^t - \mZ^\natural} ^2/ \sigma_{r}(\mZ^\natural).
\end{align*}
Combining these terms together, we have
\begin{align}
	\fronorm{\mZ^{t+1} -\mZ^\natural}  &{\leq} \left( 2 \varepsilon + \frac{4+ 3\sqrt{1+\varepsilon}}{\sigma_r(\mZ^\natural)} \fronorm{\mZ^t - \mZ^\natural}\right)\fronorm{\mZ^t - \mZ^\natural} \notag \\
	&{\leq} \lb  2\varepsilon +  \frac{4+3\sqrt{1+\varepsilon}}{16\sqrt{(1+\varepsilon)\mu_0s}} \cdot \varepsilon \rb  \cdot \fronorm{ \mZ^t -\mZ^\natural  } \label{label:(a)} \\
	&{\leq} 3\varepsilon\fronorm{ \mZ^t -\mZ^\natural  }\notag \\
	&\leq (1-\eta + 4\varepsilon)  \cdot \fronorm{ \mZ^t -\mZ^\natural  } \\
	&\leq ( 1 - \eta + \frac{3\eta}{4}) \fronorm{ \mZ^t -\mZ^\natural  }\\
	&{\leq} \frac{1}{2}\fronorm{ \mZ^t -\mZ^\natural  }\label{label:(b)},
\end{align}
where \eqref{label:(a)} follows from \eqref{condition 1} and \eqref{label:(b)} is due to $\varepsilon \leq  1/6$. Since $\fronorm{\mZ^t - \mZ^\natural}$ is a contractive sequence following from \eqref{label:(b)}, the inequality \eqref{condition 1} holds for all $t\geq 0$ by induction. Thus we complete the proof. 

\subsection{Proof of Lemma \ref{lemma:1}}
For any $\mZ\in\C^{sn_1\times n_2}$, we have
\begin{align*}
	\fronorm{\calA \calGT\calP_{\mathfrak{T}}(\mZ)}^2 &=  \la  \calA \calGT\calP_{\mathfrak{T}}(\mZ),  \calA \calGT\calP_{\mathfrak{T}}(\mZ) \ra \\
	&=\la\calGT\calP_{\mathfrak{T}}(\mZ),   \calAT\calA \calGT\calP_{\mathfrak{T}}(\mZ) \ra \\
	&= \la \mZ,    \calP_{\mathfrak{T}}\calG\lb \calAT\calA  - \calI \rb\calGT\calP_{\mathfrak{T}}(\mZ) \ra + \la \mZ, \calP_{\mathfrak{T}}  \calG\calGT\calP_{\mathfrak{T}}(\mZ) \ra\\
	&\leq  (1+\varepsilon)\fronorm{\mZ}^2,
\end{align*}
where the last inequality is due to Lemma III in \cite{chen2020vectorized}. So it follows that $\opnorm{ \calA \calGT\calP_{\mathfrak{T}}} \leq \sqrt{ 1+\varepsilon }$
and 
\begin{align*}
	\opnorm{ \calA \calGT\calP_{\mathfrak{T}_t}} &\leq \opnorm{ \calA \calGT\calP_{\mathfrak{T}}}  + \opnorm{ \calA \calGT\lb \calP_{\mathfrak{T}_t} - \calP_{\mathfrak{T}}\rb }\\
	&\leq    \sqrt{ 1+\varepsilon} + \opnorm{\calA}\cdot \opnorm{ \calP_{\mathfrak{T}_t} - \calP_{\mathfrak{T}} } \\
	&\leq \sqrt{ 1+\varepsilon} +  \frac{2\sqrt{\mu_0s}  \fronorm{\mZ^t - \mZ^\natural}  }{ \sigma_{\min}(\mZ^\natural) }\\
	&\leq 3\sqrt{ 1+\varepsilon}.
\end{align*}
Finally, a straightforward computation yields that
\begin{align}
	\opnorm{\calP_{\mathfrak{T}_t}\calG\lb  \calI - \calA^\ast\calA\rb\calGT\calP_{\mathfrak{T}_t} } &\leq \opnorm{ \lb\calP_{\mathfrak{T}_t} - \calP_{\mathfrak{T}}\rb \calG\lb  \calI - \calA^\ast\calA\rb\calGT\calP_{\mathfrak{T}_t} } + \opnorm{   \calP_{\mathfrak{T}} \calG\lb  \calI - \calA^\ast\calA\rb\calGT\calP_{\mathfrak{T}_t}}\notag\\
	& \leq \opnorm{ \lb\calP_{\mathfrak{T}_t} - \calP_{\mathfrak{T}}\rb \calG\lb  \calI - \calA^\ast\calA\rb\calGT\calP_{\mathfrak{T}_t} } \\
	&\quad + \opnorm{   \calP_{\mathfrak{T}} \calG\lb  \calI - \calA^\ast\calA\rb\calGT\lb \calP_{\mathfrak{T}_t} - \calP_{\mathfrak{T}} \rb} \notag \\
	&\quad + \opnorm{   \calP_{\mathfrak{T}} \calG\lb  \calI - \calA^\ast\calA\rb\calGT  \calP_{\mathfrak{T}} }\notag\\
	&  \leq \opnorm{ \lb\calP_{\mathfrak{T}_t} - \calP_{\mathfrak{T}}\rb \calG\calGT\calP_{\mathfrak{T}_t} }+  \opnorm{   \calP_{\mathfrak{T}} \calG\calGT\lb \calP_{\mathfrak{T}_t} - \calP_{\mathfrak{T}} \rb} \notag\\
	&\quad +\opnorm{ \lb\calP_{\mathfrak{T}_t} - \calP_{\mathfrak{T}}\rb \calG \calA^\ast\calA \calGT\calP_{\mathfrak{T}_t} }\notag\\
	& \quad+ \opnorm{   \calP_{\mathfrak{T}} \calG  \calA^\ast\calA \calGT\lb \calP_{\mathfrak{T}_t} - \calP_{\mathfrak{T}} \rb} +\opnorm{   \calP_{\mathfrak{T}} \calG\lb  \calI - \calA^\ast\calA\rb\calGT  \calP_{\mathfrak{T}} }\notag\\
	& \leq\frac{4\fronorm{\mZ^t - \mZ^\natural}  }{ \sigma_{r}(\mZ^\natural) } + \frac{2\fronorm{\mZ^t - \mZ^\natural}  }{ \sigma_{r}(\mZ^\natural) }\cdot \sqrt{\mu_0 s} \cdot \big( \opnorm{\calP_{\mathfrak{T}_t}\calG\calAT} \notag\\
	&\quad +  \opnorm{ \calP_{\mathfrak{T}}\calG\calAT} \big) + \varepsilon \label{label:(c)} \\
	&\leq \frac{4\varepsilon}{16\sqrt{(1+\varepsilon)\cdot \mu_0 s}} + \frac{8\varepsilon \sqrt{\mu_0 s(1+\varepsilon)}}{16\sqrt{(1+\varepsilon)\cdot \mu_0 s}} + \varepsilon\label{label:(d)}\\
		&\leq 2\varepsilon,\notag
\end{align}
where \eqref{label:(c)} is due to Lemma \cite[Lemma 4.1]{wei2016guarantees} and the fact that $\opnorm{\calAT} = \opnorm{\calA} \leq \sqrt{\mu_0 s}$ and $\opnorm{\calP_{\mathfrak{T}_t}\calG\calAT}  = \opnorm{\calA \calGT\calP_{\mathfrak{T}_t}} $, \eqref{label:(d)} follows from \eqref{condition 1}.

\section{Conclusion}

We propose a FIHT-VHL method to solve the blind super-resolution problem in a non-convex scheme. The convergence analysis has been established for FIHT-VHL, showing that the algorithm linearly converges to the target matrix given suitable initialization and provided the number of samples is large enough. The numerical experiments illustrate our theortical results.

\section{Acknowledgement}

We would like to thank professor Ke Wei and Yingzhou Li  for fruitful discussions.


\bibliographystyle{plain}
\bibliography{ref}

\end{document}